\documentclass[12pt,final,
thmsa,epsf,a4paper]
{article}
\usepackage{booktabs}

\newtheorem{theorem}{Theorem}

\newtheorem{corollary}{Corollary}
\newtheorem{definition}{Definition}
\newtheorem{example}{Example}
\newtheorem{lemma}{Lemma}
\newtheorem{proposition}{Proposition}
\newtheorem{remark}{Remark}

\newenvironment{proof}[1][Proof]{\emph{#1.} }{\  \hfill $\square $ \vspace{5 pt}}

\usepackage{stackrel}
\usepackage{natbib}
\usepackage{amssymb}
\usepackage[dvips]{graphics}
\usepackage{amsmath}

\usepackage{mathpazo}

\usepackage{enumerate}

\usepackage{amsfonts}

\usepackage{amssymb}

\usepackage{enumerate}

\usepackage{mathrsfs}

\usepackage[usenames]{color}

\usepackage{pgf,tikz}

\usetikzlibrary{decorations.pathreplacing,angles,quotes}
\usetikzlibrary{arrows.meta}
\usetikzlibrary{arrows, decorations.markings}

\usepackage{hyperref}
\hypersetup{
    colorlinks,
    citecolor=blue,
    filecolor=black,
    linkcolor=blue,
    urlcolor=blue
}

\usepackage{soul}
\usepackage{pgf,tikz}
\usetikzlibrary{arrows}

\usepackage[utf8]{inputenc}
\usepackage{amssymb, amsmath,geometry}
\usepackage{fancyhdr}

\usepackage{soul}
\usepackage{cancel}

\usepackage{emptypage}

\newcommand*\samethanks[1][\value{footnote}]{\footnotemark[#1]}

\begin{document}

\title{The lattice of  envy-free many-to-many matchings with contracts\thanks{We thank Jordi Massó and Alejandro Neme 
for very detailed comments. We acknowledge financial support
from UNSL through grants 032016 and 030320, and from Consejo Nacional
de Investigaciones Cient\'{\i}ficas y T\'{e}cnicas (CONICET) through grant
PIP 112-200801-00655, and from Agencia Nacional de Promoción Cient\'ifica y Tecnológica through grant PICT 2017-2355.}}


\author{Agustín G. Bonifacio\thanks{Instituto de Matem\'{a}tica Aplicada San Luis, Universidad Nacional de San
Luis and CONICET, San Luis, Argentina, and RedNIE. Emails: \texttt{abonifacio@unsl.edu.ar} (A. G. Bonifacio), \texttt{ncguinazu@unsl.edu.ar} (N. Gui\~{n}azu), \texttt{nmjuarez@unsl.edu.ar} (N. Juarez), \texttt{paneme@unsl.edu.ar} (P. Neme) and \texttt{joviedo@unsl.edu.ar} (J. Oviedo).} \and  Nadia Gui\~{n}az\'u \samethanks[2] \and Noelia Juarez\samethanks[2] \and Pablo Neme\samethanks[2] \and Jorge Oviedo\samethanks[2]}

\date{\today}
\maketitle

\begin{abstract}
We study envy-free allocations in a  many-to-many matching model with contracts in which agents on one side of the market (doctors) are endowed with substitutable choice functions and agents on the other side of the market (hospitals) are endowed with responsive preferences. Envy-freeness is a weakening of stability that allows blocking contracts involving a hospital with a vacant position and a doctor that does not envy any of the doctors that the hospital currently employs. We show that the set of envy-free allocations has a lattice structure. Furthermore, we define a Tarski operator on this lattice and  use it to model a vacancy chain dynamic process by which, starting from any envy-free allocation, a stable one is reached. 

\bigskip

\bigskip


\end{abstract}

\section{Introduction}
\label{intro}

 Models of many-to-many matching with contracts subsume as special cases many-to-many matching markets and buyer-seller markets with heterogeneous and indivisible goods.

In this paper, we study envy-free allocations in a   many-to-many matching model with contracts in which agents on one side of the market (that we call doctors) are endowed with substitutable choice functions and agents on the other side of the market (that we call hospitals) are endowed with responsive preferences. Loosely speaking, 
(doctor) envy-freeness is a weakening of stability that allows blocking contracts involving a hospital with a vacant position and a doctor that does not envy any of the doctors that the hospital currently employs.

In this setting, we study the set of envy-free allocations and prove that it has a lattice structure under a partial order by which an allocation dominates another one if, for each doctor, when the contracts of both allocations are available, she selects only the contracts of the first one. Moreover, we define a Tarski operator on this lattice and show that the set of its fixed points consists of the set of stable allocations.

The Tarski operator allows us to model vacancy chain dynamics as well. Consider a situation in which   some doctors retire and the   vacant positions generated are filled by doctors who do not have envy towards any  of their new co-workers. If the market is stable before the retirements, then it will be envy-free after.  Moreover, if our goal is to study how to restore stability, a re-equilibration process can be carried out within the envy-free allocations. To do this, we apply repeatedly the Tarski operator. In each stage, starting from an envy-free allocation, the signing of new contracts filing vacant positions (and generating new vacant positions elsewhere) produces a new envy-free allocation.  We show that the sequence of envy-free allocations thus generated  always converges to a stable allocation.

This process can be described more precisely if we also ask  doctors' choice functions to satisfy the ``law of aggregate demand''. This condition says  that when a doctor chooses from an expanded set of contracts, she  signs at least as many contracts as before.  We show two further results  under this requirement: (i) the fixed point of the Tarski operator starting from an arbitrary envy-free allocation is equal to the join of that allocation and the hospital-optimal 
 stable allocation, and (ii) the number of contracts signed by a doctor in an envy-free allocations is at most equal to the number of contracts  she signs in  any stable allocation.  
 
The use of Tarski operators is not new in the matching literature. Results on the existence of stable allocations rely on the construction of a lattice and a Tarski operator on that lattice that has as its fixed points elements that are in correspondence with  the stable allocations \citep[see][for more details]{adachi2000characterization,fleiner2003fixed,echenique2004core,hatfield2005matching,hatfield2017contract}. However, in most cases, the elements of the lattice are \emph{not} allocations themselves, and their economic interpretation is not clear. The contribution of our paper, that generalizes the approach first presented by \cite{wu2018lattice} for the many-to-one model with responsive preferences and then extended to substitutable preferences by \cite{bonifacio2021lattice}, is to provide a simpler and more economically meaningful construction.

The model and results studied in this paper encompass the many-to-one models of \cite{wu2018lattice} and \cite{bonifacio2021lattice} and the results therein. When each doctor can be assigned to at most one hospital and hospitals have responsive preferences, our model is equivalent to the one presented in \cite{wu2018lattice}. When each doctor has substitutable preferences, and each hospital can be assigned to at most one doctor, our model is equivalent to the one presented in \cite{bonifacio2021lattice} in which doctors play the role of firms and hospitals play the role of workers. 
Although the  results in this paper are generalizations of the ones presented in \cite{bonifacio2021lattice} for the many-to-one case, the fact that our setting has a  many-to-many nature  and  involves contracts  demands new techniques of proof for almost all results.

It is worth mentioning that many-to-many  markets (let alone the inclusion of contracts) are non-trivial extensions of one-to-one markets: many properties
of one-to-one markers do not extend to this wider class. Although many-to-one markets (with responsive preferences) are isomorphic
to  one-to-one markets \citep{roth1992two}, there is no
such isomorphism for many-to-many markets, even under responsive
preferences. Furthermore, (pairwise) stability is no longer equivalent in this broader framework to
other solution concepts introduced in the literature such as core stability, group
stability or setwise stability, and stable allocations may even be inefficient \citep[see][for more detailed discussions]{blair1988lattice,roth1992two,sotomayor1999three,echenique2004theory}.

The rest of the paper is organized as follows. Section \ref{seccion preliminar} presents the model and some preliminaries. All the results of the paper are presented in Section \ref{section results}: the lattice structure of the set of envy-free allocations,  the Tarski operator with its related re-equilibration process, and some further results assuming that doctors' choice functions satisfy, besides substitutability, the ``law of aggregate demand''. In Section \ref{seccion concludings} some concluding remarks are in order. Finally, an Appendix contains all the proofs for our results.

\section{Model and preliminaries}\label{seccion preliminar}
A model of many-to-many matching with contracts is specified by a set of \emph{doctors} $D$, a set of \emph{hospitals} $H$, and a set of \emph{contracts} $X$.  
For each contract $x\in X$, $x_D\in D$ is the doctor associated to $x$ and $x_H\in H$ is the hospital associated to $x.$ Given a set of contracts $Y\subseteq X$ and an agent $a\in D\cup H,$ let  $Y_a=\{x\in Y:a\in\{x_D,x_H\}\}$ be the set of contracts in $Y$ that name agent $a$.

Each doctor $d\in  D$ in endowed with a choice function $C_{d}$  that fulfills the following properties:

\begin{enumerate}[(i)]
\item $ C_{d}\left( Y\right) \subseteq Y_{d}$ and if $x,x'\in C_{d}\left(Y\right) $ and $x\neq x'$, then $x_{H}\neq x'_{H}.$

\item Substitutability: $C_{d}(Y)\cap Y' \subseteq C_{d}(Y')$ whenever $Y'\subseteq Y \subseteq X$.\footnote{This is equivalent to the following: $x\in C_d(Y)$ implies $x\in C_d(Y' \cup \{x\})$ whenever $x\in X$ and $Y'\subseteq Y \subseteq X$.}

\item Consistency:  $C_{d}(Y')=C_{d}(Y)$ whenever $C_{d}(Y)\subseteq Y'\subseteq Y$. 

\item Path independence: $C_{d}\left( Y \cup Y'\right) =C_{d}\left( C_{d}( Y)\cup Y'\right)  $ for each pair of subsets $Y$ and $Y'$ of $X$.
\end{enumerate}
Property (i) says that, given a set of contracts $Y$, the choice function of doctor $d$ selects a subset of contracts of $Y$ that name doctor $d$ and that different contracts of this choice set must name different hospitals. Property (ii) requires  the choice function to be  substitutable in the sense that no contract becomes desirable when some other contract becomes available. Property (iii) says that the choice set of a subset of $Y$ that is a superset of the choice set of $Y$ is equal to the choice set of $Y$. Lastly, Property (iv) says that the choice over a set remains the same when the set is segmented arbitrarily, the choice applied to one segment and finally the choice applied again to all chosen contracts from the segments. It is straightforward to see that Properties (ii) and (iii) imply Property (iv).

Given a set of contracts $Y$, as a consequence of Properties (i) and (iii),  since $C_d(Y)\subseteq Y_d \subseteq Y$ we have  $C_d(Y_d)=C_d(Y)$. This fact is going to be used extensively throughout the paper.

Each hospital $h\in H$ has a quota denoted by $q_h$, and a strict preference relation $\succ_h$ over all subsets of contracts in a \textbf{responsive} manner. This is formalized by the following properties:  
\begin{enumerate}
\item[(v)]for each $Y\subseteq X$ such that  $|Y_h|> q_h$,  $  \emptyset\succ_h Y_h,$
\item[(vi)] for each $Y\subseteq X$ such that $|Y_h|\leq q_h$ , each $x\in Y_h\cup \{\emptyset\}$, and each $x'\in X_h\setminus Y_h$,
$$
(Y_h\setminus \{x\})\cup \{x'\}\succ_h Y_h \text{ if and only if } x'\succ_h x.\footnote{For a thorough account of responsive preferences, see \cite{roth1992two} and references therein.} 
$$

\end{enumerate}

Property (v) says that, for each hospital,  each subset of contracts of cardinality greater than its quota is not acceptable. Lastly, Property (vi) says that, for each hospital, when a subset of contracts that name it has cardinality less than or equal to its quota, replacing any contract of this subset for a more preferred one leads to a more preferred subset of contracts.
Given preference $\succ_h$, define the  choice function  $C_h$ as follows: for each $Y\subseteq X$, $C_h(Y)$ is the most preferred subset of $Y_h$ according to $\succ_h$. Notice that $C_h$ also fulfills Conditions (i) to (iv) in the definition of $C_d$. 

A set of contracts $Y\subseteq X$ is an \textbf{allocation} if:
\begin{enumerate}[(i)]

\item for distinct $x,x' \in Y,$ $x_D \neq x'_D$ or $x_H \neq x'_H,$ and 
\item $|Y_h|\leq q_h$ for each $h\in H.$
\end{enumerate} 
The empty allocation is denoted by $Y^{\emptyset}$ and contains no contract. Furthermore, denote by $\mathcal{A}$ the set of all allocations. Condition (i) says that, in an allocation, no  doctor-hospital pair can sign more than one contract. Condition (ii) says that, in an allocation, no hospital can sign more contracts than its quota.


Given $Y\in \mathcal{A}$, we say that $Y$ is  an \textbf{individually rational} allocation if $C_{a}\left( Y\right) =Y_{a}$
for each $a\in D\cup H.$ This implies, by substitutability, that for each contract $x\in Y$ we have   $x\in C_{x_{D}}\left( \{x\} \right) $ and  $x\in C_{x_{H}}\left( \{x\} \right) .$\footnote{Notice that, by responsiveness,  $x\in C_{x_{H}}\left( \{x\} \right) $ is equivalent to $x\succ_{x_H}\emptyset$.} Denote by $\mathcal{I}$ the set of all individually rational allocations. Given $Y\in \mathcal{A}$, $x\in X \setminus Y$ is a \textbf{blocking contract} for $Y$ if
 $x\in C_{x_{D}}(Y\cup \{x\})$ and $x\in C_{x_{H}}(Y\cup \{x\})$. 
   Furthermore, allocation $Y$ is  \textbf{stable} if $Y$ is individually rational and there is no blocking contract for $Y$. Denote by $\mathcal{S}$ the set of all stable allocations.

\begin{remark}\label{remark choice}Given $Y\in \mathcal{A}$, $x\in X \setminus Y$ and $h=x_H,$ the fact that $x\in C_{h}(Y\cup \{x\})$ is,  by responsiveness,    equivalent to  $x \succ_h \emptyset$  when  $|Y_h|<q_h$, and to the existence of $x' \in Y_h$ such that  $x\succ_h x'$ when   $|Y_h|=q_h.$

\end{remark}

The key notion that we introduce in this paper is the concept of envy-freeness. In an envy-free allocation no  doctor has justified envy towards any other doctor. This is formalized in  the following definition.

\begin{definition}
Given  $Y\in \mathcal{A}$, doctor $d'$ has \textbf{justified envy}  towards doctor $d$ at $Y$ (possibly $d=d'$) if there are $x\in Y_d$ and $x'\in X_{d'}\setminus Y$ such that:\begin{enumerate}[(i)]
\item $x'_H =x_H=h,$
\item $x' \succ_{h} x$ and $x' \in C_{d'}(Y \cup \{x'\})$.
 
\end{enumerate}
An allocation $ Y $ is \textbf{(doctor) envy-free} if it is individually rational and no doctor has justified envy at $Y$. 
\end{definition}
When a doctor has justified envy towards another doctor in an allocation there are two contracts involved: one in the allocation, that names the  doctor who is envied; and another one outside the allocation, that names the  envious doctor.  Condition (i) says that these two contracts name the same hospital. Condition (ii) says that the named  hospital prefers the contract that names the envious doctor over the contract that names the envied doctor, and that the envious doctor wishes to sign the new contract with that hospital. Envy-freeness, therefore, is a weakening of stability that allows blocking contracts involving a hospital with a vacant position and a doctor that has no justified envy  towards  any doctor that the hospital currently employs. This is, if $Y$ is an envy-free allocation and $x$ is a blocking contract for $Y$, then $x\in C_{x_D}(Y\cup \{x\})$, $|Y_{x_H}|< q_{x_H}$, $x\succ_{x_H}\emptyset$, and $x'\succ_{x_H} x$ for each $x'\in Y$. 
Denote by $\mathcal{E}$ the set of all envy-free allocations.

\section{Results}\label{section results}

In this section, we present our results. First, we show that the set of envy-free allocations has a lattice structure. To do this, we define a partial order and a binary operation within the set of envy-free allocations and prove that this binary operation is the join with respect to that partial order. Second, we define a Tarski operator on the lattice of envy-free allocations and show that has the stable allocations as the set of fixed points. This operator models a re-equilibration process that, starting from any envy-free allocation, leads to a stable one. Moreover, in the final subsection, we present some results that follows from requiring  that doctors' choice functions fulfill, in addition, the law of aggregate demand. All proofs are relegated to the Appendix.

\subsection{Lattice structure of the set of envy-free allocations}
In a many-to-many matching model, \cite{blair1988lattice} defines a partial order  that can be generalized to a many-to-many model with contracts as follows. A set of contracts  dominates another if each doctor wishes to keep the contracts signed under the first set, even if all contracts of the  second set are also available, and do not wish to sign  any new contract. Formally, given two sets of contracts $Y,Y'\in X$, we say that \textbf{$\boldsymbol{Y}$ weakly}\textbf{ Blair-dominates $\boldsymbol{Y'}$} and write $\boldsymbol{Y \succeq_D Y'}$ when $Y=C_d\left(Y \cup Y'\right)$ for each $d\in D$.  If $Y \succeq_D Y'$ and $Y \neq Y',$  we say that \textbf{$\boldsymbol{Y$ Blair-dominates $Y'}$} and  write  $\boldsymbol{Y \succ_D Y'}.$

 Given two doctor envy-free
allocations $ Y $ and $ Y' $, we define  $ \lambda^{Y,Y'}\subseteq 2^X$ as follows:
\begin{enumerate}[(i)]
\item for each $d\in D,$ $\lambda^{Y,Y'} _d=C_d\left(Y\cup Y'\right),$
\item for each $h \in H,$ $\lambda^{Y,Y'} _h=\{x\in \bigcup_{d\in D} \lambda^{Y,Y'}_d:x_H=h  \}.$
\end{enumerate}
\noindent   Under $\lambda^{Y,Y'},$ (i) doctors want to  sign the best subset of contracts  among those signed by them in either allocation, and (ii) hospitals agree with doctors  that want to sign a contract with  them. Notice that by item (i),  if $Y_d=Y_d'=\emptyset$, then $\lambda^{Y,Y'}_d=\emptyset.$ The following proposition shows that  $\lambda^{Y,Y'}$ is an envy-free allocation, and that  it is actually the \emph{join} of $Y$ and $Y'$.\footnote{Given a partially ordered set $(\mathcal{L},\succeq)$, and two elements $x,y\in \mathcal{L}$, an element $z\in \mathcal{L}$ is an \textit{upper bound} of $x$ and $y$ if $z\succeq x$ and $z\succeq y$. An element $w\in \mathcal{L}$ is the \textit{join} of $x$ and $y$ if and only if (i) $w$ is an upper bound of $x$ and $y$, and (ii) $t\succeq w$ for each upper bound $t$ of $x$ and $y$. The definitions of \textit{lower bound} and \textit{meet} of $x$ and $y$ are dual and we omit them.}

\begin{proposition}\label{lemma lambda in E} 
Let $Y$ and $Y'$ be two distinct envy-free allocations. Then,
\begin{enumerate}[(i)]
\item $\lambda^{Y,Y'} $ is a  envy-free allocation.
\item $\lambda^{Y,Y'}$ is the join of $Y$ and $Y'$ under the partial order $\succeq_D$.
\end{enumerate}

\end{proposition}

From now on, given two envy-free allocations $Y$ and $Y'$, we denote $\lambda^{Y,Y'}$ as $Y\vee Y'$.
Now we are in a position to present the main theorem of this subsection. This result  states that the set of envy-free allocations has a lattice structure with respect to Blair's partial order.

\begin{theorem}\label{teorema lattice}
The set of envy-free allocations is a lattice under the partial order $\succeq_D$.
\end{theorem}

In the following example, given two envy-free allocations, we show how to compute the join between them.  

\begin{example}\label{ejemplo 1}
Consider a market in which $D=\{d_{1},d_{2}\}$ is the set of doctors and $H=\{h_{1},h_{2}\}$ is the set of hospitals with quotas $q_{h_{i}}=2,$ $i=1,2$. A contract $x$ that is signed between doctor $d_{i}$ and hospital $h_{j}$ is denote by $x_{ij}.$ 
The hospitals have the following responsive preferences:
\[
\begin{array}{|c||c|}\hline
\succ _{h_{1}} & \succ _{h_{2}}\\ \hline\hline 
x_{11},y_{11}&x_{12},y_{12} \\
y_{11},y_{21}&y_{12},y_{22} \\
x_{21},y_{11}&x_{12},y_{22} \\
x_{11},y_{21}&x_{22},y_{12} \\
x_{11},x_{21}&x_{12},x_{22} \\
x_{21},y_{21}&x_{22},y_{22} \\
y_{11}&y_{12} \\
x_{11}&x_{12}\\
y_{21}&y_{22} \\
x_{21} &x_{22} \\ 
\hline
\end{array}%
\]%

The doctors are endowed with the following choice functions:\[
\begin{array}{|c|c||c|c|}\hline
X_{d_{1}}& C_{d_{1}} & X_{d_{2}}& C_{d_{2}}\\ \hline \hline
x_{11},x_{12},y_{11},y_{12} & x_{11},x_{12}&x_{21},x_{22},y_{21},y_{22} & x_{21},x_{22}\\
x_{11},y_{11},y_{12} & x_{11}& x_{21},y_{21},y_{22} & x_{21},y_{22}\\
x_{12},y_{11},y_{12} & x_{12},y_{11}&x_{22},y_{21},y_{22} & x_{22},y_{21}\\
x_{11},x_{12},y_{11} & x_{11},x_{12}&x_{21},x_{22},y_{21} & x_{21},x_{22}\\
x_{11},x_{12},y_{12} &x_{11},x_{12}&x_{21},x_{22},y_{22} &x_{21},x_{22}\\
x_{11},x_{12} & x_{11},x_{12}&x_{21},x_{22} & x_{21},x_{22}\\
x_{11},y_{11} & x_{11}&x_{21},y_{21} & x_{21}\\
x_{11},y_{12} & x_{11}&x_{21},y_{22} & x_{21},y_{22}\\
x_{12},y_{11} &  x_{12},y_{11}& x_{22},y_{21} &  x_{22},y_{21}\\
x_{12},y_{12} &  x_{12}& x_{22},y_{22} &  x_{22}\\
y_{11},y_{12} &  y_{11},y_{12}& y_{21},y_{22} &  y_{21},y_{22}\\
x_{11} & x_{11}& x_{21} & x_{21}\\
x_{12} & x_{12}&x_{22} & x_{22}\\
y_{11} & y_{11}&y_{21} & y_{21}\\
y_{12} &  y_{12}& y_{22} &  y_{22}\\\hline
\end{array}%
\]
where, for instance, $X_{d_{1}}$ are all the possible subsets of contracts that name doctor $d_1$. For the subset of contracts $Y=\{x_{11},x_{12},y_{11},y_{12}\}$, $C_{d_1}(Y)= \{x_{11},x_{12}\}$. It is easy to see that the choice functions satisfy substitutability and consistency.

\noindent Notice that this market has twelve envy-free allocations, and only four of them are also  stable allocations. Now, consider the following two envy-free allocations: $Y=\{y_{11},y_{12},y_{21}\}$ and $Y'=\{y_{11},y_{12},y_{22}\}$. In allocation $Y=\{y_{11},y_{12},y_{21}\}$  $d_1$  signs  two contracts, one with $h_1$ ($y_{11}$) and the other one with $h_2$ ($y_{12}$) and $d_2$  signs   contract $y_{21}$ with $h_1$. Since $C_{d_1}(\{y_{11},y_{12}\}\cup \{y_{11},y_{12}\})=\{y_{11},y_{12}\}$ and $C_{d_2}(\{y_{21}\}\cup \{y_{22}\})=\{y_{21},y_{22}\}$, $Y\vee Y'=\{y_{11},y_{12},y_{21},y_{22}\}.$ It is easy to see that $Y\vee Y'\in \mathcal{E}.$ Moreover, $Y\vee Y'\in \mathcal{S}.$
In Figure \ref{reticulado1} we present all envy-free allocations and its lattice structure. Stable allocations are depicted in boldface.
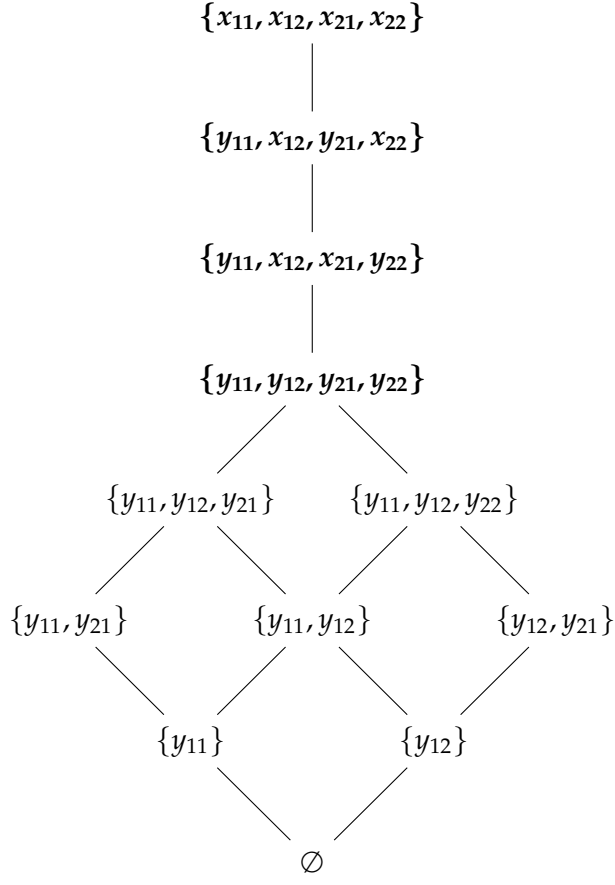
\begin{figure}[h!]
\begin{center}
\begin{tikzpicture}[scale=0.8]

\node (0) at (0,0) {\small{$\emptyset$}};
\node (1) at (-2,2) {\small{$\{y_{11}\}$}};
\node (2) at (2,2) {\small{$\{y_{12}\}$}};
\node (3) at (-4,4) {\small{$\{y_{11},y_{21}\}$}};
\node (4) at (0,4) {\small{$\{y_{11},y_{12}\}$}};
\node (5) at (4,4) {\small{$\{y_{12},y_{21}\}$}};
\node (6) at (-2,6) {\small{$\{y_{11},y_{12},y_{21}\}$}};
\node (7) at (2,6) {\small{$\{y_{11},y_{12},y_{22}\}$}};
\node (8) at (0,8) {\small{$\boldsymbol{\{y_{11},y_{12},y_{21},y_{22}\}}$}};
\node (9) at (0,10) {\small{$\boldsymbol{\{y_{11},x_{12},x_{21},y_{22}\}}$}};
\node (10) at (0,12) {\small{$\boldsymbol{\{y_{11},x_{12},y_{21},x_{22}\}}$}};
\node (11) at (0,14) {\small{$\boldsymbol{\{x_{11},x_{12},x_{21},x_{22}\}}$}};

\draw (0) to (1);
\draw (0) to (2);
\draw (1) to (3);
\draw (1) to (4);

\draw (2) to (4);
\draw (2) to (5);

\draw (3) to (6);
\draw (4) to (6);
\draw (4) to (7);
\draw (5) to (7);
\draw (6) to (8);
\draw (7) to (8);
\draw (8) to (9);
\draw (9) to (10);
\draw (10) to (11);

\end{tikzpicture}
\caption{The lattice of Example \ref{ejemplo 1}.}
\label{reticulado1}
\end{center}
\end{figure}

\end{example}

\subsection{Re-equilibration process}

In this subsection, we define a Tarski operator in the envy-free lattice that
describes a possible re-equilibration process. This process models how, starting from
a worker-quasi-stable matching, a decentralized sequence of offers in which unemployed
workers are hired and cause new unemployments, produces a sequence of
worker-quasi-stable matchings that converges to a stable matching. In the first subsection,
we present the operator, show some of its properties, and prove that the set of
its fixed points is the set of stable matchings. In the second subsection, we discuss the
re-equilibration process, based on our Tarski operator, that models a vacancy chain that
leads towards a stable matching.

For each $Y\in \mathcal{E}$, define the following sets:

$$
\mathcal{B}^Y=\{ x \in X\setminus Y :x \text{ is a blocking contract for } Y\}
$$

$$
\mathcal{B}^Y_{\star}=\{ x\in \mathcal{B}^Y: \text{ there is no } x' \in \mathcal{B}^Y \text{ such that } x'\succ _{h} x  \text{ where } h=x_H=x'_H \}
$$
and for each $d\in D$, 
$$
\mathcal{B}_d ^Y=\{x\in \mathcal{B}^Y_{\star} : x_D=d\}.
$$

Given an envy-free allocation $Y$,  $\mathcal{B}^Y$ denotes the set of all blocking contracts of $Y$. $\mathcal{B}^Y_{\star}$ contains, for each hospital, the most preferred contract of $\mathcal{B}^Y$ among those that name it. Finally, $\mathcal{B}^Y_d$ contains all contracts from $\mathcal{B}^Y_{\star}$ that name doctor $d$. Note that for each $h\in H$, there is at most one contract in $\mathcal{B}^Y_{\star}$ that names it.

Our Tarski operator $\mathcal{T}: \mathcal{E} \to \mathcal{E}$ assigns to each envy-free allocation $Y \in \mathcal{E}$ another envy-free allocation $\mathcal{T}^Y \in \mathcal{E}$ and is defined as follows:
\begin{enumerate}[(i)]
\item for each $d\in D$, $\mathcal{T}^Y_d=C_d \left(Y\cup \mathcal{B}_d^Y \right)$  
\item for each $h\in H$,  $\mathcal{T}^Y_h=\{x\in \mathcal{T}^Y:x_H=h\}$.
\end{enumerate}

The following theorem presents some good properties of operator $\mathcal{T}$. It states that: (i) the image of an envy-free allocation under the operator is also envy-free, (ii) the operator is weakly Pareto improving for the doctors, and (iii) its set of fixed points is the set of stable allocations.   

\begin{theorem}\label{Proposition punto fijo}
For any  $Y \in \mathcal{E}$, the following hold:
\begin{enumerate}[(i)]
\item $\mathcal{T}^Y \in \mathcal{E},$  
\item $\mathcal{T}^Y\succeq_D Y,$ 
\item $\mathcal{T}^Y=Y$ if and only if $Y \in \mathcal{S}.$ 
\end{enumerate}
\end{theorem}

Another important property of operator $\mathcal{T}$ is its isotonicity. This is, $\mathcal{T}^Y \succeq_D \mathcal{T}^{Y'}$ for envy-free allocations $Y,Y'$ such that $Y' \succeq_D Y'$.\footnote{This property is proven in Lemma \ref{T isotone} in the Appendix.}  This facts toghether with Tarski's fixed point theorem\footnote{Remember that Tarski's theorem \citep{tarski1955lattice} states that if $(\mathcal{L},\geq )$ is a complete lattice and $T:\mathcal{L}\longrightarrow\mathcal{L}$ is isotone,  then the set of fixed points of $T$ is non-empty and forms a complete lattice with respect to $\geq$.} and Theorem \ref{Proposition punto fijo} (iii) allow us to state the following result.

\begin{theorem}\label{teorema S es no vacio y lattice}
The set of stable allocations is non-empty and forms a lattice with respect to the partial order $\succeq_D.$
\end{theorem}

We can further observe that the repeated application of  operator $\mathcal{T}$ can be interpreted as a vacancy chain dynamic within envy-free allocations as follows.
Consider a situation in which hospitals have vacant positions (for instance, after the retirement of some doctors) and no doctor envies the position filled by  another doctor, i.e. and envy-free allocation $Y$. Applying operator $\mathcal{T}$ to the aforementioned allocation involves the following steps:
\begin{enumerate}[(i)]
\item  Among the possible blocking contracts of the envy-free allocation, each hospital selects its most preferred one ($\mathcal{B}_{\star}^Y$ defined in the previous subsection).
\item  Each doctor selects the most preferred subset of contracts among those blocking contracts selected in the previous step that name her and her currently signed contracts (generating the new envy-free allocation $\mathcal{T}^Y$).
\item Once each doctor signs the new contracts, either  some hospitals have new vacant positions $(\mathcal{T}^Y \in \mathcal{E}\setminus \mathcal{S})$ or no new blocking contract can be formed $(\mathcal{T}^Y \in \mathcal{S})$.
\end{enumerate}

If the $\mathcal{T}^Y$ is stable, the process ends reaching a stable allocation. If $\mathcal{T}^Y$ is not stable, the process continues applying operator $\mathcal{T}$ to the envy-free allocation $\mathcal{T}^Y$. The sequence of allocations generated by this process belongs to the set of envy-free allocations by Theorem \ref{Proposition punto fijo} (i) and each allocation in the sequence Pareto improves (for the doctors) upon the previous allocation in the sequence by Theorem \ref{Proposition punto fijo} (ii). This process continues until, by finiteness, it reaches a fixed point that turns out to be a stable allocation by Theorem \ref{Proposition punto fijo} (iii).

In order to prove existence of stable allocations, both  \cite{hatfield2005matching} and \cite{hatfield2017contract} define a lattice over the cartesian product of contracts $X \times X.$ On this lattice, \cite{hatfield2017contract} define a Tarski operator and show that the fixed points of such operator are in one-to-one correspondence with the stable allocations, while \cite{hatfield2005matching} define a similar operator without such one-to-one correspondence. \cite{hatfield2005matching} also study vacancy chain dynamics by means of their Tarski operator. However, we believe that our description of this process in terms of (envy-free) allocations instead of ordered pairs of contracts is simpler and has a clearer economic interpretation.

\subsection{Further results with \textit{LAD}}

In this subsection, by requiring an additional condition on doctors' choice functions, we can describe more accurately the re-equilibration process by means of the lattice structure of the set of  envy-free allocations. This additional condition is the ``law of aggregate demand", that says that when a doctor chooses from an expanded set of contracts, she  signs at least as many contracts as before. Formally, 
\begin{definition}
Choice function $C_d$  satisfies the\textbf{ law of aggregate demand (\textit{LAD})} if $Y'\subseteq Y\subseteq X$ implies $|C_d(Y')|\leq |C_d(Y)|.$
\end{definition}
We know that, starting from an envy-free allocation and iterating our operator $\mathcal{T}$, we reach a fixed point of $\mathcal{T}$. Assuming \textit{LAD}, the lattice structure can help us to identify this fixed point: it is the join of the original envy-free allocation and the hospital-optimal stable allocation $Y^H$.\footnote{The set of stable allocations under substitutable choice functions is very well-structured. It contains two distinctive stable allocations: the doctor-optimal stable allocation $Y^D$ and the hospital-optimal stable allocation $Y^H$. The  allocations $Y^H$ is unanimously considered by all hospitals to be the best among all stable allocations and by
all doctors to be the pessimal stable allocation \citep[see][for more details]{hatfield2017contract}. Analogous opposition of interests results have been identified in most matching settings, including those of \cite{roth1984evolution,blair1988lattice,hatfield2005matching,echenique2004theory}. }
To formally present this result, for $Y \in \mathcal{E}$, let $\mathcal{F}^Y$ denote the fixed point of $\mathcal{T}$  obtained by iterating it starting at allocation $Y$.

\begin{theorem}\label{thLAD}
 Let $Y$ be an envy-free allocation.  If  doctors' choice functions satisfy \textit{LAD}, then  $\mathcal{F}^Y=Y \vee Y^{H}.$
\end{theorem}

The following corollary presents an important feature of the structure of the set of envy-free allocations. It states that any envy-free allocation $Y$ that Blair-dominates $Y^H$ is actually a stable allocation. This happens because as $Y$ Blair-dominates $Y^H$ the join between them is $Y$ and, at the same time, that join is  equal to $\mathcal{F}^Y$ by Theorem \ref{thLAD}. Therefore, $Y=\mathcal{F}^Y$. Moreover, by Theorem \ref{Proposition punto fijo} (iii), $\mathcal{F}^Y$ is a stable allocation. 

\begin{corollary}\label{corolario LAD}
Let $Y$ be an envy-free allocation. If  doctors' choice functions satisfy \textit{LAD} and $Y \succeq _{D} Y^{H}$, then $Y$ is a stable allocation.
\end{corollary}

Since, $Y^H$ is the doctor-pessimal stable allocation, the implications of this corollary are twofold:  (i) an envy-free allocation that Blair-dominates \textit{any} stable allocation is also stable, and (ii) an  allocation that is envy-free but not stable is  either Blair-incomparable to or  Blair-dominated by $Y^H.$

The following proposition states that, for each doctor, the amount of contracts signed in an envy-free allocation is always less or equal to the amount of contracts signed in any stable allocation. The proof of this proposition uses the rural hospitals theorem that, adapted to the setting of many-to-many matching markets with contracts, states that when preferences are substitutable and satisfy  \textit{LAD}, each agent signs the same number of contracts at every stable allocation \citep[see Theorem 4 in  Appendix B in][]{hatfield2017contract}.

\begin{proposition}\label{proposition con LAD}
Let $Y$ be an envy-free allocation and let $Y'$ be a stable allocation. If doctors' choice functions satisfy \textit{LAD}, $|Y_d|\leq |Y'_d|$ for each $d\in D.$
\end{proposition}

The following example shows that the requirement of \textit{LAD} is necessary for all the results of this subsection to hold.

\begin{example}
Consider a market in which $D=\{d_{1},d_{2}\}$ is the set of doctors and $H=\{h_{1},h_{2},h_{3}\}$ is the set of hospitals with quotas $q_{h_{i}}=1,$ $i=1,2,3$. The hospitals have the following preferences:

\[
\begin{array}{|c||c||c|}\hline
\succ _{h_{1}} & \succ _{h_{2}}&\succ _{h_{3}} \\ \hline
x_{21}& x_{22}&x_{13}\\
x_{11} &x_{12}&x_{23} \\ 
\hline
\end{array}
\]

\noindent The doctors are endowed with the following choice functions:\[
\begin{array}{|c|c||c|c|}\hline
X_{d_{1}}& C_{d_{1}} & X_{d_{2}}& C_{d_{2}}\\ \hline \hline
x_{11},x_{12},x_{13} & x_{11},x_{12}& x_{21},x_{22},x_{23} & x_{23}\\
x_{11},x_{12} & x_{11},x_{12}& x_{21},x_{22} & x_{21},x_{22}\\
x_{11},x_{13} & x_{11}&x_{21},x_{23} & x_{23}\\
x_{12},x_{13}  & x_{12}&x_{22},x_{23}  & x_{23}\\
x_{11} &x_{11}&x_{21} &x_{21}\\
x_{12} & x_{12}&x_{22} & x_{22}\\
x_{13}  & x_{13}& x_{23}  & x_{23}\\\hline
\end{array}%
\]

\noindent It is easy to see that $\succ_{d_i}$ is substitutable for  $i=1,2$, but $\succ_{d_2}$ does not fulfill \textit{LAD}. Consider the following sets of contracts $\{x_{21},x_{22},x_{23}\}$ and $\{x_{21},x_{22}\}$.  We can observe that $|C_{d_2}(\{x_{21},x_{22},x_{23}\})|=|\{x_{23}\}|<|\{x_{21},x_{22}\}|=|C_{d_2}(\{x_{21},x_{22}\})|$ contradicting the definition of \textit{LAD}. Now, consider the following allocations: $Y^D=\{x_{11},x_{12},x_{23}\},$ $Y=\{x_{11}, x_{23}\},$ $Y^H=\{x_{13},x_{21},x_{22}\}$ and $Y'=\{x_{21},x_{22}\}$. Note that $Y,Y'\in \mathcal{E}\setminus \mathcal{S}$ (contract $x_{12}$ blocks 
$Y$ and contract $x_{13}$ blocks $Y'$), and $Y^D$ and $Y^H$ are the doctor-optimal and the hospital-optimal stable allocations respectively. Now, we can make the following observations: (i)  $Y\vee Y^H=Y$  showing that Theorem \ref{thLAD} does not hold without \textit{LAD}. (ii) $Y^D\succeq_D Y \succeq_D Y^H\succeq_D Y'$ showing that Corollary \ref{corolario LAD} does not hold without \textit{LAD}. (iii) Consider allocations $Y^D$ and $Y'$, and doctor $d_2$. Then, $|\{x_{23}\}|=|Y^D_{d_2}|<|Y'_{d_2}|=|\{x_{21},x_{22}\}|$ showing that Proposition \ref{proposition con LAD} does not hold without \textit{LAD}.
\end{example}

\section{Concluding remarks}\label{seccion concludings}
The main motivation of this paper was to provide a framework to study a re-equilibration process for the most general two-sided model in which stability can be guaranteed (substitutable many-to-many matching with contracts). We accomplish this goal by introducing the concept of envy-free allocation that in itself has a meaningful economic interpretation.

The full set of stable allocations has been computed by means of several algorithms in the matching literature, starting from the one-to-one model all through the many-to-many model with contracts \citep[see][among others]{mcvitie1971stable,irving1986complexity,martinez2004algorithm,dworczak2021deferred,pepa2022matching,bonifacio2022cycles}. An interesting future line of research is the development of an algorithm to compute the full set of envy-free allocations.

\section*{Appendix}

\noindent \begin{proof}[Proof of Proposition \ref{lemma lambda in E}] Let  $Y, Y' \in \mathcal{E}.$ 

\noindent (i) In order to see that $\lambda^{Y,Y'} \in \mathcal{E},$ we proceed in three steps.

\noindent \textbf{Step 1:} $\boldsymbol{\lambda^{Y,Y'} \in \mathcal{A}}.$  Let $ x,x'\in \lambda^{Y,Y'}$ with $ x\neq x' $ and let $d \in D.$ By definition of $\lambda^{Y,Y'} $ and propoerty (i) of $C_d $ we have  $ x_H \neq x'_H.$ It remains to see that $|\lambda^{Y,Y'}_h|\leq q_h$ for each $h\in H$.  Assume there is $h\in H$ such that $|\lambda^{Y,Y'}_h|>q_{h}$. Since $Y$ and $Y'$ are  allocations, we have that $|Y_{h}|\leq q_{h}$ and $|Y_{h}'|\leq q_{h}$. Therefore, there are $x,x'\in X$  such that $x\in \lambda^{Y,Y'}_h\setminus Y_{h}'$ and $x'\in \lambda^{Y,Y'}_h\setminus Y_{h}$. Moreover, as $\lambda^{Y,Y'}_h \subseteq \left(Y_h\cup Y'_{h}\right)\setminus Y_{h}',$ it follows that $x \in Y_h \setminus Y_h'.$ Similarly, $x'\in Y'_h\setminus Y_{h}$. Hence,  $x\neq x'.$ Let $d=x_D$ and $d'=x'_D.$  Since $x\in\lambda^{Y,Y'}_d=C_d(Y \cup Y')$ and $x \notin Y'_h$, by substitutability, \begin{equation}\label{ecu 0 lemma lambda in E}
x\in C_{d}(Y' \cup \{x\}).
\end{equation}
If $x \succ_h x',$ then \eqref{ecu 0 lemma lambda in E} implies that doctor $d$ has justified envy towards doctor $d'$ at $Y',$ contradicting that $Y' \in \mathcal{E}.$ Thus,  
 \begin{equation}\label{ecu 1 lemma lambda in E}
 x'\succ _{h}x.  
\end{equation}  Since $x'\in\lambda^{Y,Y'}_{d'}=C_d(Y \cup Y')$ and $x' \notin Y_h$, by substitutability, \begin{equation}\label{ecu 0 lemma lambda in E bis}
x'\in C_{d'}(Y \cup \{x'\}).
\end{equation}
Now, \eqref{ecu 1 lemma lambda in E} and \eqref{ecu 0 lemma lambda in E bis} imply that doctor $d'$ has justifed envy towards doctor $d$ at $Y,$ contradicting that $Y \in \mathcal{E}.$ Therefore, $|\lambda^{Y,Y'}_h|\leq q_{h}.$ We conclude that $\lambda^{Y,Y'} \in \mathcal{A}.$ 

\noindent \textbf{Step 2:} $\boldsymbol{\lambda^{Y,Y'} \in \mathcal{I}}.$  Let $d\in D.$ By definition of $\lambda^{Y,Y'}$ and path-independence of $C_d,$ 
\begin{equation}\label{ecu 2 lemma lambda in E}
C_{d}\left( \lambda^{Y,Y'} \right) =C_{d}\left( C_{d}(Y\cup Y')\right)=C_{d}(Y\cup Y')=\lambda^{Y,Y'}_d.
\end{equation}
Now, take any $h\in H$ and any $x\in \lambda^{Y,Y'}_h$. By definition of $\lambda^{Y,Y'}$, $x\in Y \cup Y'$. Hence, $x\in Y$ or $x\in  Y'$. Since both $Y$ and $Y'$ are individually rational allocations, $x\succ _{h}\emptyset.$ Then, by responsiveness, $C_{h}\left( \lambda^{Y,Y'} \right) =\lambda^{Y,Y'}_h.$ This fact together with \eqref{ecu 2 lemma lambda in E} proves that  $\lambda^{Y,Y'} \in \mathcal{I}.$

\noindent \textbf{Step 3:} $\boldsymbol{\lambda^{Y,Y'} \in \mathcal{E}}.$  Assume this is not the case.  Then, there is  a doctor $d'\in D$ that has justified envy towards a doctor $d \in D$ (possibly $d=d'$) at $\lambda^{Y,Y'}.$ This implies that  there are $x\in \lambda^{Y,Y'}_d$ and $x'\in X_{d'}\setminus \lambda^{Y,Y'}$ such that 
$x'_H =x_H=h,$
\begin{equation}\label{ecu 4 lemma lambda in E}
x' \succ_{h} x\text{ and }x' \in C_{d'}(\lambda^{Y,Y'} \cup \{x'\}).
\end{equation}
 By definition of $\lambda^{Y,Y'} $, $x'\in C_{d'}\left(
C_{d'}\left( Y\cup Y^{\prime }\right) \cup \{x'\}\right)$ and, by path independence,   
\begin{equation}\label{ecu 5 lemma lambda in E}
x'\in
C_{d'}\left( Y\cup Y^{\prime }\cup \{x'\}\right).
\end{equation}  
Notice that $x' \notin Y \cup Y'.$ Otherwise, \eqref{ecu 5 lemma lambda in E} implies $x' \in  \lambda^{Y,Y'},$ contradicting our hypothesis. Since $x\in \lambda^{Y,Y'}\subseteq Y\cup Y'$, assume w.l.o.g. that $x\in Y.$ By \eqref{ecu 5 lemma lambda in E} and  the substitutability of $C_{d'}$, 
\begin{equation}\label{x in C_h(YUx) y x in C_h(Y'Ux)}
x'\in C_{d'}(Y\cup \{x'\}).
\end{equation}
By \eqref{ecu 4 lemma lambda in E}, $x' \succ_{h} x.$ Therefore, by \eqref{x in C_h(YUx) y x in C_h(Y'Ux)},  doctor $d'$ has justified envy  towards doctor $d$ at $Y$ since. This contradicts that $Y \in \mathcal{E}.$ We conclude that  $\lambda^{Y,Y'} \in \mathcal{E}.$

\noindent (ii) In order to see that $\lambda^{Y,Y'}$ is the join of $Y$ and $Y'$ under the partial order $\succeq_D$, we proceed as follows. By Proposition \ref{lemma lambda in E} (i), $\lambda^{Y,Y'} \in \mathcal{E}.$ First, we prove that $\lambda^{Y,Y'}$ is an upper bound of $Y$ and $Y'$ for the doctors. By definition of $\lambda^{Y,Y'}$ and path independence, 
$$
C_{d}(\lambda^{Y,Y'} \cup Y)= C_{d}(C_{d}(Y \cup Y') \cup Y)=
C_{d}(Y\cup Y' \cup Y)= C_{d}(Y \cup Y')=\lambda^{Y,Y'}_d
$$
for each $d\in D$.  This implies that  $\lambda^{Y,Y'}\succeq_D Y.$ Similarly, $\lambda^{Y,Y'}\succeq_D Y'.$ Second, we prove that $\lambda^{Y,Y'}$ is the join of $Y$ and $Y'$ for the doctors. Let $\overline{Y} \in \mathcal{E}$ such that $\overline{Y}\succeq_D Y$ and $\overline{Y}\succeq_D Y'.$ That is,
\begin{equation}\label{ecu 1 lemma lambda es join}
\overline{Y}_d=C_{d}(\overline{Y} \cup Y) \text{ and } \overline{Y}_d=C_{d}(\overline{Y} \cup Y')
\end{equation}
for each $d\in D.$ We need to show that $ \overline{Y} \succeq_D \lambda^{Y,Y'}$, that is $\overline{Y}_d=C_{d}(\overline{Y} \cup \lambda^{Y,Y'})$ for each $d\in D$. Using repeatedly path independence, \eqref{ecu 1 lemma lambda es join}, and the definition of $\lambda^{Y,Y'}$, 
$$
\overline{Y}_d=C_{d}(\overline{Y} \cup Y)=C_{d}(C_{d}(\overline{Y} \cup Y') \cup Y)=C_{d}(\overline{Y} \cup Y' \cup Y)=$$
$$C_{d}(\overline{Y} \cup C_{d}( Y' \cup Y))=C_{d}(\overline{Y} \cup \lambda^{Y,Y'})
$$
for each $d\in D.$ Thus, $ \overline{Y} \succeq_D \lambda^{Y,Y'}$. Therefore, $\lambda^{Y,Y'}$ is the join for $Y$ and $Y'$.
\end{proof}

\noindent \begin{proof}[Proof of Theorem \ref{teorema lattice}]
First, by Proposition \ref{lemma lambda in E} (ii), the set of envy-free allocations $\mathcal{E}$ forms a join-semilattice under the partial order $\succeq_D$.\footnote{A partially ordered set $\mathcal{L}$ is called a \textit{join-semilattice} if any two elements in $\mathcal{L}$ have a join. If any two elements in  $\mathcal{L}$ also have a meet, then $\mathcal{L}$ is called a \textit{lattice} \citep[see][for more details]{stanley2011enumerative}.} Second, the  empty allocation $Y^{\emptyset}$ in which all hospitals have their positions unfilled (that is by definition an envy-free allocation) is the minimum element of $\mathcal{E}$ under the partial order $\succeq_D$. To see this, let $Y\in\mathcal{E}$. Since $Y_d=C_d(Y \cup Y^{\emptyset})$ for each $d\in D$, it follows that $Y=Y \vee Y^{\emptyset}$. Thus, $Y \succeq_D Y^{\emptyset}$ for each $Y \in \mathcal{E}.$ 
Finally, given that the set of envy-free allocations is finite and is a join-semilattice with a minimum element, it follows that the set of envy-free allocations forms a lattice under the partial order $\succeq_D$  \citep[see][for more details]{stanley2011enumerative}.
\end{proof}

\noindent \begin{proof}[Proof of Theorem \ref{Proposition punto fijo}] Let $Y \in \mathcal{E}.$
\begin{enumerate}[(i)]
\item   In order to see that $\mathcal{T}^Y \in \mathcal{E},$ we proceed in three steps.  

\noindent \textbf{Step 1: $\boldsymbol{\mathcal{T}^Y \in \mathcal{A}}$}.  We need to show that 
\begin{equation}\label{allocation 1}
 \text{ for distinct }x,x' \in \mathcal{T}^Y,x_D \neq x'_D \text{  or  }x_H \neq x'_H,
\end{equation}
and 
\begin{equation}\label{allocation 2}
|\mathcal{T}^Y_h|\leq q_h\text{ for each }h\in H.
\end{equation}
If there are distinct $x,x'\in \mathcal{T}^Y$  such that $x_D \neq x'_D$ we are done, so assume $x_D=x'_D=d$. Since  $x,x'  \in C_d \left(Y\cup \mathcal{B}_d^Y \right),$ by definition of $C_d$, $x_H\neq x'_H.$  This proves \eqref{allocation 1}. 
To see \eqref{allocation 2}, assume otherwise. Then, there is $h\in H$ such that $|\mathcal{T}^Y_h|> q_h.$  Since $Y \in \mathcal{A}$, we have that $|Y_h|\leq q_h.$ Hence, there is  $x\in \mathcal{T}^Y_h \setminus Y_h$. Let $d=x_D.$ Then $x \in \mathcal{T}^Y_d$ and, as $\mathcal{T}^Y_d=C_d \left(Y\cup \mathcal{B}_d^Y \right) \subseteq Y\cup \mathcal{B}_d^Y,$  $x \in   \mathcal{B}_d^Y$ because $x\notin Y.$  Hence, $x$ is a blocking contract for $Y$, i.e.\begin{equation}\label{ecu 1 teorema T allocation}
  x \in C_d \left(Y\cup \lbrace x \rbrace \right) \text{ and }  x \in C_h \left(Y\cup \lbrace x \rbrace \right)
 \end{equation}
Now, we claim  $ |Y_{h}|< q_{h}.$ Assume otherwise that $ |Y_{h}|= q_{h} $. Since $x$ is a blocking contract for $Y,$ there is $x' \in Y  $ with $ x_H=x'_H=h $ such that $  \left(Y_h \setminus \{x'\}\right) \cup \{x\} \succ_h Y_h,$  and, by responsiveness, this happens if and only if 
  \begin{equation}\label{ecu 2 teorema T allocation}
     x\succ_h x'.
  \end{equation}
  Thus, by \eqref{ecu 1 teorema T allocation} and \eqref{ecu 2 teorema T allocation} $ d $ has justified envy towards doctor $x'_D $ at $ Y $. This contradicts that $ Y \in \mathcal{E} $. Then,  $ |Y_{h}|< q_{h},$ proving our claim.  Furthermore, since $ |Y_{h}|< q_{h} $ and we assume that $|\mathcal{T}^Y_h|> q_h$, there is $\widetilde{x} \in X$ such that $\widetilde{x}\neq x$ and $\widetilde{x}\in \mathcal{T}^Y_h \setminus Y_h.$ Let $d'=\widetilde{x}_D.$ Then $\widetilde{x} \in \mathcal{T}^Y_{d'}$ and, as $\mathcal{T}^Y_{d'}=C_{d'} \left(Y\cup \mathcal{B}_{d'}^Y \right) \subseteq Y\cup \mathcal{B}_{d'}^Y,$  $\widetilde{x} \in   \mathcal{B}_{d'}^Y$ because $\widetilde{x}\notin Y.$  Remember that $x \in   \mathcal{B}_d^Y$, thus $x\succ_h y$ for each blocking contract $y$  for $Y$ such that $h=y_H$. In particular, $x\succ_h \widetilde{x}.$ Then $\widetilde{x} \notin \mathcal{B}_{\star}^Y,$ contradicting that  $\widetilde{x}\in\mathcal{B}^Y_{d'}.$  Therefore, \eqref{allocation 2} holds. We conclude that $\mathcal{T}^Y \in \mathcal{A}.$ 

\textbf{Step 2: $\boldsymbol{\mathcal{T}^Y \in \mathcal{I}}$.}  Assume otherwise, then there is $x\in \mathcal{T}^Y$ such that $x\notin C_d(\mathcal{T}^Y)$ or $x\notin C_h(\mathcal{T}^Y)$ where $d=x_D$ and $h=x_H$. By definition of $\mathcal{T}^Y$ and path independence, $\mathcal{T}^Y_{d}=C_d \left(Y \cup \mathcal{B}^Y_{d}\right)=C_d \left(C_d \left(Y \cup \mathcal{B}^Y_{d}\right)\right)=C_d(\mathcal{T}^Y).$ Hence, $x\in \mathcal{T}^Y_d$ implies $x \in C_d(\mathcal{T}^Y)$ and, by our contradiction hypothesis,  $x\notin C_h(\mathcal{T}^Y).$  Since $C_h(\mathcal{T}^Y) \subseteq \mathcal{T}^Y_h \subseteq \mathcal{T}^Y$, by consistency of $C_h$ it follows that $C_h(\mathcal{T}^Y_h)=C_h(\mathcal{T}^Y).$ Hence,  $x\notin C_h(\mathcal{T}^Y_h).$  As $|\mathcal{T}^Y_h \setminus \{x\}| < q_h,$  by Remark \eqref{remark choice}, $x\notin C_h(\mathcal{T}^Y_h)=C_h(\left(\mathcal{T}^Y_h\setminus \{x\}\right) \cup \{x\})$ is equivalent to 
\begin{equation}\label{step 2}
\emptyset \succ_h x.
\end{equation}
The fact that $Y \in \mathcal{I}$ and \eqref{step 2} imply that $x\notin Y.$ Moreover, as $\mathcal{T}^Y_d= C_d\left(Y \cup \mathcal{B}^Y_{d}\right)\subseteq Y \cup \mathcal{B}^Y_{d},$ $x \in \mathcal{T}^Y_d$ implies  $x\in \mathcal{B}^Y_{d}$ and, therefore, $x$ is a blocking contract for $Y$. Thus, $x\in C_h (Y\cup \{x\}).$ Notice that, by consistency of $C_h,$ we have $C_h (Y\cup \{x\})=C_h (Y_h\cup \{x\}),$ so  $x\in C_h (Y_h \cup \{x\}).$ There are two cases to consider. If $|Y_h|= q_h,$  by Remark \ref{remark choice}, there is a contract $y\in Y$ such that $x\succ_h y$. As $Y \in \mathcal{I}$ we have $y\succ_h \emptyset$ and, by transitivity, $x\succ_h  \emptyset.$ If $|Y_h|< q_h$,  by Remark \ref{remark choice}, we  also obtain $x\succ_h \emptyset.$  In either case, we contradict \eqref{step 2}. We conclude that $\mathcal{T}^Y \in \mathcal{I}.$

\textbf{Step 3: $\boldsymbol{\mathcal{T}^Y \in \mathcal{E}}$}.  Assume otherwise. Then, there are $x \in \mathcal{T}^Y$ and $x' \in X \setminus \mathcal{T}^Y$ with $x_D=d,$ $x'_D=d',$  and $x_H=x'_H=h$ such that $x'\succ_h x$ and $x' \in C_{d'}(\mathcal{T}^Y \cup \{x'\}).$ By definition of $\mathcal{T}$ and path-independence, $C_{d'}(\mathcal{T}^Y \cup \{x'\})= C_{d'}(C_{d'}(Y \cup \mathcal{B}_{d'}^Y) \cup \{x'\})= C_{d'}(Y \cup \mathcal{B}_{d'}^Y \cup \{x'\}).$ Therefore, 
\begin{equation}\label{step 3-1}
x' \in C_{d'}(Y \cup \mathcal{B}_{d'}^Y \cup \{x'\}).
\end{equation}
If $x' \in Y,$ then  $x' \in C_{d'}(Y \cup \mathcal{B}_{d'}^Y)=\mathcal{T}^Y_{d'}.$ This contradicts that $x' \in X \setminus \mathcal{T}^Y.$ Hence, 
\begin{equation}\label{step 3-2}
x' \in X \setminus Y. 
\end{equation}
Furthermore, \eqref{step 3-1} and  substitutability imply 
\begin{equation}\label{T in E eq 1}
x' \in C_{d'}(Y \cup \{x'\}). 
\end{equation}
As $\mathcal{T}^Y_{d}=C_{d}(Y \cup \mathcal{B}^Y_{d})\subseteq Y \cup \mathcal{B}^Y_{d}$ and $x \in \mathcal{T}^Y_{d},$ we have  $x \in Y \cup \mathcal{B}^Y_{d}.$ There are two cases to consider:
\begin{enumerate}
\item[$\boldsymbol{1}.$] $\boldsymbol{x \in Y}.$ Then, \eqref{step 3-2} and  \eqref{T in E eq 1} together with the facts that  $x_D=d,$ $x'_D=d',$  $x_H=x'_H=h,$ and  $x'\succ_h x$  imply that $d'$ has justified envy towards $d$ at $Y,$  contradicting that $Y \in \mathcal{E}.$

\item[$\boldsymbol{2}.$] $\boldsymbol{x \in \mathcal{B}^Y_d}.$ Then,  $x \in \mathcal{B}^Y_\star.$ This implies  that $x$ is a blocking contract for $Y$ and, as $Y \in \mathcal{E},$ $|Y_h| < q_h.$  As $\mathcal{T}^Y \in \mathcal{I}$  and $x \in \mathcal{T}^Y,$ $x\succ_h \emptyset.$ Moreover, as $x' \succ_h x$ by this step's hypothesis, it follows that $x' \succ_h \emptyset.$ This last fact together with \eqref{T in E eq 1} and $|Y_h| < q_h$ imply that $x' \in \mathcal{B}^Y.$ Then,  $x' \succ_h x$ contradicts that $x \in \mathcal{B}^Y_\star.$
\end{enumerate}
In either case we reach a contradiction. We conclude that $\mathcal{T}^Y \in \mathcal{E}.$

\item  Let $d \in D.$ 
 By definition of $\mathcal{T}$ and path-independence, 
$$C_d(\mathcal{T}^Y_d \cup Y)=C_d(C_d(Y \cup \mathcal{B}^Y_d)\cup Y)=C_d(Y \cup \mathcal{B}^Y_d\cup Y)=C_d(Y \cup \mathcal{B}^Y_d)=\mathcal{T}_d^Y.$$ Thus, $C_d(\mathcal{T}^Y_d \cup Y)=\mathcal{T}^Y_d.$ As $d$ is arbitrary, $\mathcal{T}^Y \succeq_D Y.$  

\item  ($\Longrightarrow$) Assume $Y \in \mathcal{E} \setminus \mathcal{S}.$ Thus, $\mathcal{B}^Y_\star \neq \emptyset.$ Therefore, there is $d \in D$ such that $\mathcal{B}^Y_d \neq \emptyset.$ This implies that $\mathcal{T}^Y_d=C_d(Y \cup \mathcal{B}^Y_d)\neq Y_d.$ Hence, $\mathcal{T}^Y \neq Y.$ 

\noindent ($\Longleftarrow$) Assume $Y \in \mathcal{S}.$ Thus, $\mathcal{B}^Y=\emptyset.$  Therefore, $\mathcal{B}^Y_d=\emptyset$ for each $d \in D.$ Then, by definition of $\mathcal{T},$ $\mathcal{T}^Y_d=Y_d$ for each $d \in D.$ Hence, $\mathcal{T}^Y=Y.$
\end{enumerate}
\end{proof}

In order to prove Theorem \ref{teorema S es no vacio y lattice}, we first prove the following lemma.
\begin{lemma}\label{T isotone}
If $Y$ and $Y'$ are two doctor envy-free allocation such that $ Y \succeq_D Y' $, then $ \mathcal{T}^{Y} \succeq_D \mathcal{T}^{Y'}$.
\end{lemma}
\begin{proof}
Let $Y,Y'\in \mathcal{E}$ be such that $ Y \succeq_D Y' $ and assume that $ \mathcal{T}^{Y} \succeq_D \mathcal{T}^{Y'}$ does not hold. This implies the existence of $ d\in D $ such that
\begin{equation}\label{ecu 1 lemma T isotone}
\mathcal{T}_d^Y\neq C_d \left( \mathcal{T}^{Y} \cup  \mathcal{T}^{Y'}\right).
\end{equation}
Using the definition of $ \mathcal{T}^{Y} $ and  path-independence,
\begin{equation}\label{ecu 2 lemma T isotone}
C_d \left( \mathcal{T}^{Y} \cup  \mathcal{T}^{Y'}\right)= C_{d} \left( C_d\left(Y \cup \mathcal{B}^Y_d \right)  \cup  C_d\left(Y' \cup \mathcal{B}^{Y'}_d \right)\right)
 \end{equation}
 $$=C_{d} \left(Y \cup \mathcal{B}^Y_d  \cup  C_d\left(Y' \cup \mathcal{B}^{Y'}_d \right)\right)=
C_{d} \left(Y \cup \mathcal{B}^Y_d  \cup  Y' \cup \mathcal{B}^{Y'}_d \right) $$
$$=C_{d} \left(Y \cup Y'  \cup \mathcal{B}^Y_d\cup \mathcal{B}^{Y'}_d \right). $$
Using again path-independence, it follows that 
\begin{equation}\label{ecu 3 lemma T isotone}
C_{d} \left(Y \cup Y' \cup \mathcal{B}^Y_d\cup \mathcal{B}^{Y'}_d \right)=C_{d} \left(C_d\left(Y \cup Y'\right) \cup  \mathcal{B}^Y_d\cup \mathcal{B}^{Y'}_d \right)
\end{equation}
as by hypothesis $ C_d\left(Y \cup Y'\right)=Y $, using   \eqref{ecu 2 lemma T isotone} and \eqref{ecu 3 lemma T isotone} we get
\begin{equation}
C_d \left( \mathcal{T}^{Y} \cup  \mathcal{T}^{Y'}\right)=C_{d} \left(Y\cup \mathcal{B}^Y_d\cup \mathcal{B}^{Y'}_d \right)
\end{equation}

Now, using the definition of $ \mathcal{T} $ and \eqref{ecu 1 lemma T isotone}, \eqref{ecu 3 lemma T isotone} becomes
\begin{equation}\label{ecu 4 lemma T isotone}
C_d\left(Y\cup\mathcal{B}^Y_d \right)\neq C_{d} \left(Y\cup \mathcal{B}^Y_d\cup \mathcal{B}^{Y'}_d \right).
\end{equation}
By \eqref{ecu 4 lemma T isotone}, there is a contract $ x\in X $ such that
\begin{equation}\label{ecu 5 lemma T isotone}
x\in C_d\left(Y\cup \mathcal{B}^Y_d\cup \mathcal{B}^{Y'}_d \right)
\end{equation}
and
\begin{equation}\label{ecu 6 lemma T isotone}
x\in \mathcal{B}^{Y'}_d\setminus \left(Y\cup \mathcal{B}^Y_d \right). 
\end{equation}
 Since $ Y'\in \mathcal{E} $ and $ Y' $ has a blocking contract $x,$  $|Y'_h| < q_h$ with $ h=x_H $.  There are two cases to consider:
 
 \begin{enumerate}
 
 \item[$\boldsymbol{1}$.] $\boldsymbol{|Y_h|=q_h}.$ Let $ x'\in Y_h \setminus Y'_h$ and let $d'=x'_D$. Notice that, as  $ Y \succeq_D Y' $, $ x'\in Y =C_{d'}\left(Y \cup Y'\right).$ This, together with $ x'\in Y_h \setminus Y'_h$ implies, by substitutability, that
 \begin{equation}\label{ecu 7 lemma T isotone}
  x'\in C_{d'}\left(Y' \cup  \left\lbrace x'\right\rbrace   \right).
\end{equation}  
 By \eqref{ecu 6 lemma T isotone}, $x \notin Y.$ This fact together with $x' \in Y$ imply $x' \neq x.$ Furthermore, by \eqref{ecu 6 lemma T isotone}, 
\begin{equation}\label{ecu 8 lemma T isotone}
x \succ_h x'.
\end{equation} 
Now, by \eqref{ecu 5 lemma T isotone}, \eqref{ecu 6 lemma T isotone}, and substitutability,
 \begin{equation}\label{ecu 9 lemma T isotone}
  x \in C_{d}\left(Y \cup  \left\lbrace x\right\rbrace   \right).
\end{equation}  
 Since $x' \in Y,$ $x \notin Y,$ \eqref{ecu 8 lemma T isotone}, and \eqref{ecu 9 lemma T isotone}, it follows that  $ d $ has justified envy towards $d' $ at $ Y $, contradicting that $ Y\in \mathcal{E}. $
 
 \item[$\boldsymbol{2}$.] $\boldsymbol{|Y_h|<q_h}.$  By \eqref{ecu 5 lemma T isotone}, \eqref{ecu 6 lemma T isotone}, $|Y_h|<q_h,$ and substitutability, it follows that $x \in C_d(Y \cup \{x\})$ and $x \succ_h \emptyset.$ Therefore, $x \in \mathcal{B}^Y.$  Since, by \eqref{ecu 6 lemma T isotone}, $x \notin \mathcal{B}^Y_d,$ it follows that there is $x' \in \mathcal{B}^Y$ such that $x' \succ_h x.$ Let $d'=x'_D.$ Then,
\begin{equation}\label{ecu 10 lemma T isotone}
x' \in C_{d'}(Y \cup \{x'\}).
\end{equation} 
 By consistency, $C_{d'}(Y\cup \{x'\})=C_{d'}(Y_{d'}\cup \{x'\}).$ As $Y \succeq_D Y',$ $Y_{d'}=C_{d'}(Y\cup Y').$ Therefore, using also  path independence, $$C_{d'}(Y\cup \{x'\})=C_{d'}(C_{d'}(Y\cup Y')\cup \{x'\})=C_{d'}(Y\cup Y' \cup \{x'\}).$$ Then, $x' \in C_{d'}(Y\cup \{x'\})$ implies $x' \in C_{d'}(Y\cup Y'\cup \{x'\})$ and in turn, by substitutability, 
\begin{equation}\label{ecu 11 lemma T isotone}
x' \in C_{d'}(Y' \cup \{x'\}). 
\end{equation}  
 As $|Y_h'|<q_h,$ by \eqref{ecu 11 lemma T isotone} we have that $x' \in \mathcal{B}^{Y'}.$ Moreover, as $x' \succ_h x$ we have  $x \notin \mathcal{B}^{Y'}_{\star},$ contradicting \eqref{ecu 6 lemma T isotone}.  
 \end{enumerate}
 As in each case we reach a contradiction, we conclude that $\mathcal{T}^Y \succeq_D \mathcal{T}^{Y'}.$ 
\end{proof}

\noindent \begin{proof}[Proof of Theorem \ref{teorema S es no vacio y lattice}]
We need to see that operator $\mathcal{T}$ verifies the hypothesis of Tarski's theorem.  First, notice that the  envy-free lattice is finite and therefore complete. Second, by Theorem \ref{Proposition punto fijo} (i), $\mathcal{T}$  maps the envy-free lattice to itself. Finally, $\mathcal{T}$ is isotone by Lemma \ref{T isotone}. Then, by Tarski's theorem, the set of fixed points of $\mathcal{T}$ is non-empty and forms a lattice under $\succeq_D$. Moreover, by Theorem \ref{Proposition punto fijo} (iii), the set of fixed points of operator $\mathcal{T}$ is the set of stable allocations. 
\end{proof}

The following two lemmata are needed to prove Theorem \ref{thLAD}.
\begin{lemma}\label{lemma suma lo mismo las cardinalidades con LAD}
Let $Y$ be an allocation. Then, $\sum_{h\in H}|Y_h|=\sum_{d\in D}|Y_d|.$
\end{lemma}
\begin{proof}
Let $Y$ be an allocation. Recall that $Y_d$ is the subset of contracts in $Y$ that name doctor $d$ and, similarly, $Y_h$ is the subset of contracts in $Y$ that name hospital  $h$. Therefore, $\sum_{h\in H}|Y_h|=|Y|=\sum_{d\in D}|Y_d|.$
\end{proof}

\begin{lemma}\label{V es estable}
Let $Y$ be an envy-free allocation and $Y'$ be a stable allocation. Then, $Y\vee Y'$ is a stable allocation.
\end{lemma}
\begin{proof}
Let $Y$ be an envy-free allocation and $Y'$ be a stable allocation. By, Proposition \ref{lemma lambda in E} (i),  $Y\vee Y' \in \mathcal{E}$. Assume that $Y\vee Y' \notin \mathcal{S}.$ Thus,  there is a blocking contract ${x}$ for $Y\vee Y'$. Let  ${h}={x}_H$ and $d={x}_D.$ Then,  ${x}\in C_{{d}}(Y\vee Y' \cup \{{x}\})$, ${x}\in C_{{h}}(Y\vee Y' \cup \{{x}\}),$ and
\begin{equation}\label{ecu 1 teorema el V es estale}
|(Y\vee Y')_{{h}}|<q_{{h}}.
\end{equation} By definition of $Y\vee Y'$ and path independence, $C_{{d}}(Y\vee Y' \cup \{{x}\})=C_{{d}}(C_d(Y\cup Y') \cup \{{x}\})= C_{{d}}(Y \cup Y' \cup \{{x}\}).$  Thus, $x\in C_{{d}}(Y \cup Y' \cup \{{x}\}).$ Furthermore, by substitutability, 
\begin{equation}\label{ecu 2 teorema el V es estale}
{x}\in C_{{d}}(Y' \cup \{{x}\}).
\end{equation}
Since $x$ is a blocking contract for $Y\vee Y'$, Remark \ref{remark choice} and \eqref{ecu 1 teorema el V es estale} imply ${x}\succ_{{h}} \emptyset$. This last fact together with the stability of allocation $Y'$ and \eqref{ecu 2 teorema el V es estale} imply that 
\begin{equation}\label{ecu 3 teorema el V es estale}
|Y'_{{h}}|=q_{h}.
\end{equation}
 (Otherwise, ${x}$ is a blocking contract for $Y'$ which is a contradiction). Next, we claim that there is $\widetilde{h} \in H$ such that 
\begin{equation}\label{ecu 3.5 teorema el V es estale}
|(Y\vee Y')_{\widetilde{h}}|>| Y'_{\widetilde{h}}|.
\end{equation} Assume that \eqref{ecu 3.5 teorema el V es estale} does not hold. Then, $|(Y\vee Y')_{h}|\leq| Y'_{h}|$ for each $h\in H$. Notice that, by \eqref{ecu 1 teorema el V es estale} and  \eqref{ecu 3 teorema el V es estale}, for ${h}\in H$ we have $|(Y\vee Y')_{{h}}|<| Y'_{{h}}|.$ 
This implies that $\sum_{h\in H}|(Y\vee Y')_{h}|<\sum_{h\in H}| Y'_{h}|.$ 
Thus, by Lemma \ref{lemma suma lo mismo las cardinalidades con LAD}, 
\begin{equation}\label{ecu 4 teorema el V es estale}
\sum_{d\in D}|(Y\vee Y')_{d}|<\sum_{d\in D}| Y'_{d}|. 
\end{equation}
Furthermore, since $Y' \subseteq Y \cup Y'$ we have, by \textit{LAD} and the fact that $Y'\in \mathcal{I}$, $|Y'_d|=|C_d(Y')|\leq |C_d(Y \cup Y')|$. This implies that, by definition of $Y\vee Y'$, $\sum_{d\in D}| Y'_{d}|\leq \sum_{d\in D}| (Y\vee Y')_{d}|.$ This contradicts \eqref{ecu 4 teorema el V es estale}, implying that \eqref{ecu 3.5 teorema el V es estale} holds and the claim is proven. Thus, there is $\widetilde{h} \in H$ such that $|(Y\vee Y')_{\widetilde{h}}|>| Y'_{\widetilde{h}}|.$ Then, there is a contract $\widetilde{x}\in (Y\vee Y') \setminus Y'$ such that $\widetilde{x}_H=\widetilde{h}.$ Let $\widetilde{d}=\widetilde{x}_D.$ Since $Y\vee Y'\in \mathcal{I},$ $ Y\vee Y'=C_{\widetilde{d}}(Y\vee Y')=C_{\widetilde{d}}(Y\vee Y' \cup\{\widetilde{x}\}).$ Therefore, $x\in Y\vee Y'$ implies $x\in C_{\widetilde{d}}(Y\vee Y' \cup\{\widetilde{x}\}).$ Furthermore, by definition of $Y\vee Y'$ and path independence,  $\widetilde{x}\in C_{\widetilde{d}}(Y \cup Y' \cup\{\widetilde{x}\})$ and, by substitutability, 
\begin{equation}\label{ecu 5 teorema el V es estale}
\widetilde{x}\in C_{\widetilde{d}}(Y' \cup\{\widetilde{x}\}).
\end{equation} 
Given that $Y\vee Y'$ is an allocation, by \eqref{ecu 3.5 teorema el V es estale} we have that  $|Y'_{\widetilde{h}}|<|(Y\vee Y')_{\widetilde{h}}|\leq q_{\widetilde{h}}$. Thus, $|Y'_{\widetilde{h}}|< q_{\widetilde{h}}$, $\widetilde{x} \notin Y'$, and $\widetilde{x}\succ_{\widetilde{h}} \emptyset$ imply, by Remark \ref{remark choice}, that
\begin{equation}\label{ecu 6 teorema el V es estale}
\widetilde{x}\in C_{\widetilde{h}}(Y' \cup\{\widetilde{x}\}).
\end{equation} Then, \eqref{ecu 5 teorema el V es estale} and \eqref{ecu 6 teorema el V es estale} imply that   $\widetilde{x}$ is a blocking contract for $Y'$, contradicting the stability of allocation $Y'$. Therefore,  $Y\vee Y'$ is a stable allocation.
\end{proof}

\noindent \begin{proof}[Proof of Theorem \ref{thLAD}]
Let $Y \in \mathcal{E}.$ By Lemma \ref{V es estable}, $Y \vee Y^{H}$ is a stable allocation. By Lemma \ref{T isotone} and Theorem \ref{Proposition punto fijo} (iii),  $Y \vee Y^{H} \succeq _{D} Y$ implies that $Y \vee Y^{H} \succeq_	D  \mathcal{F}^Y$. Moreover, given that $\mathcal{T}$ is a weakly Pareto improving operator by Theorem \ref{Proposition punto fijo} (ii), $\mathcal{F}^Y \succeq_D Y$. Since $Y^{H}$ is the doctors' pessimal stable allocation, $\mathcal{F}^Y \succeq_D Y^{H}.$ As $\mathcal{F}^Y \succeq_D Y$ and  $\mathcal{F}^Y \succeq_D Y^{H},$ by definition of join,  $\mathcal{F}^Y \succeq _D Y \vee Y^{H}$. Thus, by antisymmetry, $\mathcal{F}^Y = Y \vee Y^{H}.$
\end{proof}

\noindent \begin{proof}[Proof of Proposition \ref{proposition con LAD}]
Let $Y \in \mathcal{E}$ and $d \in D.$ By \emph{LAD} and individual rationality of $Y,$ we have that $|\mathcal{T}^Y_{d}|=|C_d(Y \cup \mathcal{B}^Y_{d})|\geq|C_d(Y)| = |Y_d|.$ Iterating we obtain that $|\mathcal{F}^Y_d| \geq |Y_d|.$ By definition,  $\mathcal{F}^Y$ is a fixed point of our Tarski operator. Thus, $\mathcal{F}^Y$ is stable by Theorem \ref{Proposition punto fijo} (iii). Then, by the Rural Hospital Theorem  
 $|\mathcal{F}^Y_d|=|Y'_d|$ for each $Y' \in \mathcal{S}$ and each $d \in D.$ Therefore,  $|Y_d|\leq|\mathcal{F}^Y_d|=|Y'_d|$ for each $Y' \in \mathcal{S}$ and each $d \in D.$ 
\end{proof}


\end{document}